\UseRawInputEncoding
\documentclass[11pt,letter]{article}

\usepackage{amsmath,amssymb,amsthm}
\usepackage{bussproofs}
\usepackage{fullpage}
\usepackage{authblk}

\usepackage[pdfstartview=FitH,
            pdfpagemode=None,
            backref,
            colorlinks=true]{hyperref}
            \hypersetup{urlcolor=[rgb]{0.8, .2, .5},
                         linkcolor=[rgb]{0.1,0.1,0.6},
                         citecolor=[rgb]{0.2,0.2,0.7}
            }
\usepackage{complexityToC} %  some basic complexity classes
\usepackage{bbm} % for math align etc.
\usepackage{nicefrac} % better looking fraction symbol etc.
\usepackage{stmaryrd}
\usepackage{thmtools,thm-restate}
\usepackage{tikz} % figures
\usepackage{xspace}
\usepackage{setspace}
\usepackage{mdframed}
\usepackage{basicmathmacros} % Iddo  basic math macros, like symbols, operators etc.
\date{}
\newtheorem{theorem}{Theorem}
\newtheorem{lemma}[theorem]{Lemma}
\newtheorem{corollary}[theorem]{Corollary}

\newtheorem{claim}[theorem]{Claim}
\newtheorem{definition}[theorem]{Definition}
\newtheorem*{theorem*}{Theorem}
\newtheorem*{lemma*}{Lemma}
\newtheorem*{corollary*}{Corollary}
\newtheorem*{proposition*}{Proposition}
\newtheorem*{claim*}{Claim}
\newtheorem*{definition*}{Definition}

\newcommand{\hardz}{\ensuremath{\sum_{i,j,k,\ell\in[n]} z_{ijk\ell}x_ix_j x_k x_\ell - \beta=0}}
\newcommand{\hardzp}{\ensuremath{\sum_{i,j,k,\ell\in[n]} z_{ijk\ell}x_ix_j x_k x_\ell - \beta}}

\newcommand{\sml}{\mathrm{sml}}

\newcommand{\relrk}{\mathrm{relrk}}

\newcommand{\ks}{\mathsf{ks}}
\newcommand{\ak}{\lfloor \alpha k \rfloor}

\newcommand{\ZZ}{\mathbb{Z}}

\usepackage{color}
\usepackage[normalem]{ulem}
\usepackage{graphicx}

\newcounter{enumresume}

\begin{document}

\sloppy

\title{Simple Hard Instances for Low-Depth Algebraic Proofs\thanks{This project has received funding from the European Research Council (ERC) under the European Union's
Horizon 2020 research and innovation programme (grant agreement No 101002742)".}}
%%%%%%%%%%%%%%%%%%%%%%%%%%%%%%%%%%%%%%%%%%%

\newcommand\extrafootertext[1]{%
    \bgroup
    \renewcommand\thefootnote{\fnsymbol{footnote}}%
    \renewcommand\thempfootnote{\fnsymbol{mpfootnote}}%
    \footnotetext[0]{#1}%
    \egroup
}

% Example:
%
% \author[1]{Author A}
% \author[1]{Author B}
% \author[1]{Author C}
% \author[2]{Author D}
% \author[2]{Author E}
% \affil[1]{Department of Computer Science, \LaTeX\ University}
% \affil[2]{Department of Mechanical Engineering, \LaTeX\ University}
% \affil[ ]{\textit {\{email1,email2,email3,email4,email5\}@xyz.edu}}
%

% Setting for authorblk package: determine what comes % between authors.
%
 \renewcommand\Authsep{\qquad }
 \renewcommand\Authand{ }
 \renewcommand\Authands{\qquad }

%
% %
 \author{Nashlen Govindasamy\thanks{E-mail: nashlen.govindasamy@gmail.com}}
 \author{Tuomas Hakoniemi\thanks{E-mail: t.hakoniemi@imperial.ac.uk}}
 \author{Iddo Tzameret\thanks{E-mail: iddo.tzameret@gmail.com. \url{https://www.doc.ic.ac.uk/\~itzamere/}}}
 \affil{Imperial College London\\ {\normalsize Department of Computing}}

% make the title area

\maketitle

\begin{abstract}
We prove super-polynomial lower bounds on the size of propositional proof systems operating with constant-depth algebraic circuits over fields of zero characteristic. Specifically, we show that the  subset-sum variant $\hardz$, for Boolean variables, does not have polynomial-size IPS refutations where the refutations are multilinear and written as constant-depth circuits.

Andrews and Forbes (\begin{footnotesize}STOC\end{footnotesize}'22) established recently a constant-depth IPS lower bound, but their hard instance does not have itself small constant-depth circuits, while our instance is computable already with small depth-2 circuits.

Our argument relies on extending the recent breakthrough lower bounds against constant-depth algebraic circuits by Limaye, Srinivasan and Tavenas (\begin{footnotesize}FOCS\end{footnotesize}'21) to the functional lower bound framework of Forbes, Shpilka, Tzameret and Wigderson (\begin{footnotesize}ToC\end{footnotesize}'21), and may be of independent interest. Specifically,  we construct a polynomial $f$ computable with small-size constant-depth circuits, such that  the multilinear polynomial computing $\nicefrac{1}{f}$ over Boolean values and its  \mbox{appropriate set-multilinear} projection are hard for  constant-depth circuits.
\end{abstract}

% 2) AF22 uses a reduction while we use directly the %this  Our result is stronger in the sense that s the first unconditional super polynomial lower bound against proof systems operating with constant-depth algebraic circuits, in terms of their algebraic circuit complexity and for instances with small (polynomial) magnitude coefficients (prior unrestricted IPS lower bounds were either conditional \cite{AGHT20}, or for instances with coefficients of exponential-magnitude and consequently in terms of the bit-size of the refutations in contrast to their algebraic circuit size \cite{Ale20}, or for much more restricted variants of IPS \cite{FSTW21}).

%%%%%%%%%%%%%%%%%%%%%%%%%%%%%%%%%%%%%%%%%%%%%
%%%%%%%%%%%%%%%%%%%%%%%%%%%%%%%%%%%%%%%%%%%%%
\section{Introduction}
%%%%%%%%%%%%%%%%%%%%%%%%%%%%%%%%%%%%%%%%%%%%%
%%%%%%%%%%%%%%%%%%%%%%%%%%%%%%%%%%%%%%%%%%%%%

%% Proof Complexity intro

Proof complexity predominantly aims to establish lower bounds on proof size in different proof systems.
From the perspective of complexity theory this can be viewed as the dual goal to circuit complexity. While circuit complexity aims to prove lower bounds on minimal circuit size required to decide membership in certain languages, e.g., SAT, proof complexity aims to establish lower bounds on the minimal size of proofs witnessing membership in a certain language, e.g., UNSAT, the language of unsatisfiable Boolean formulas. Here a proof is simply a witness that can be checked efficiently.

In circuit complexity it is usual to consider different or restricted types of circuits (e.g., constant-depth circuits or constant-depth circuits with counting gates modulo a prime). Similarly, in proof complexity it is standard to consider different or restricted types of proof systems, for example proofs using a prescribed set of inference rules that derive clauses from existing ones (i.e., resolution).

The ideal goal of circuit complexity is to prove that the class \NP\ is different from the class of languages decidable by polynomial-size circuits (and hence $\NP \neq \P/\poly$). Similarly, the  overarching view of proof complexity is that of an attempt to prove lower bounds on stronger and stronger proof systems in the hope to get as close as possible to ruling out the existence of  any  proof system that admits short proofs for all unsatisfiable formulas (namely, membership in UNSAT;  and similarly for proving membership in other important languages). A language that admits no short efficiently verifiable proofs is by definition outside the class \NP, from which we conclude, in the case of ruling out short proofs of  UNSAT, that $\NP\neq\coNP$ (and hence  $\NP\neq\P$). This view of proof complexity is usually called \emph{The Cook's Programme} of proof complexity.

%%%%%%%%%%%%%%%%%%%%%%%%%%%%%%%%%%%%%%%%%%%%%
%%%%%%%%%%%%%%%%%%%%%%%%%%%%%%%%%%%%%%%%%%%%%
\subsection{Algebraic Proof Systems}
%%%%%%%%%%%%%%%%%%%%%%%%%%%%%%%%%%%%%%%%%%%%%
%%%%%%%%%%%%%%%%%%%%%%%%%%%%%%%%%%%%%%%%%%%%%

One important strand of proof complexity deals with \emph{algebraic proof systems} of increasing strength. Algebraic proof systems prove that a set of multivariate polynomials do not have a common 0-1 root over a field. The arguably canonical algebraic proof system is the relatively weak Polynomial Calculus (PC for short) \cite{CEI96}
 in which proofs start from a set of polynomial equations, and proceed to add and multiply existing polynomials until one reaches the unsatisfiable equation $1=0$ (proving that the initial polynomials do not have a common 0-1 solution). The ``static'' version of the polynomial calculus is called Nullstellensatz \cite{BeameIKPP96}. In Nullstellensatz a proof of the unsatisfiability of a set of axioms, given as polynomial equations  $\{f_i(\vx)=0\}$ over a field, is simply a \emph{single} polynomial combination of the axioms that equals 1 as a formal  polynomial, namely:  
\begin{equation}
\label{eq:NS-first}
\sum_i g_i(\vx)\cd f_i(\vx) = 1\,,
\end{equation}
for some polynomials $\{g_i(\vx)\}$ (it is said to be \emph{static} because the proof is given as a single polynomial combination instead of deriving 1 ``dynamically'' step-by-step as in PC).

The size-complexity of proofs in both PC and Nullstellensatz is sparsity, namely the total number of monomials in all the polynomials appearing along the proof. The sparsity  measure is what makes these proof systems weak (e.g., even a simple proof-line like $(x_1-1)\cdots(x_n-1)=0$ accounts for an exponential size   because the number of monomials in it is $2^n$).

While counting the total number of monomials in algebraic proofs towards their size-complexity yields comparatively weak proof systems, it is natural to think of stronger algebraic proofs by representing polynomials in a more compact manner than sparsity. In particular, one can consider writing polynomials using algebraic circuits. This idea has circulated in proof complexity starting from Pitassi \cite{Pit97,Pit98}, and subsequently in Grigoriev and Hirsch \cite{GH03}, Raz and Tzameret \cite{RT06,RT07,Tza11-I&C}, and finally in the introduction of the Ideal Proof System (IPS) by Grochow and Pitassi \cite{GP18} which loosely speaking is  the  Nullstellensatz proof system in which proofs are written as algebraic circuits (indeed, \cite{FSTW21} showed that IPS is equivalent to Nullstellensatz in which the polynomials $g_i$ in \autoref{eq:NS-first} are written as algebraic circuits). 

Accordingly, it is natural to consider proof systems that sit  between the weak Nullstellensatz on the one end and the strong IPS on the other end. This is done  by writing polynomials in proofs with restricted kind of algebraic circuits, such as constant-depth circuits \cite{GH03,GP18,IMP20,AF22}, noncommutative formulas \cite{Tza11-I&C,LTW18,Tza11-I&C}, algebraic branching programs \cite{Tza11-I&C,FSTW21,Kno17}, multilinear formulas \cite{RT06,RT07,FSTW21} and very recently algebraic proofs with additional extension variables over large fields \cite{Ale21} or finite fields \cite{IMP22}.
%(roughly simulating a restricted use of circuits in proofs) \cite{IMP22}.

%%%%%%%%%%%%%%%%%%%%%%%%%%%%%%%%%%%%%%%%%%%%%
%%%%%%%%%%%%%%%%%%%%%%%%%%%%%%%%%%%%%%%%%%%%%
\subsection{State of the Art in Algebraic Proof-Size Lower Bounds}
%%%%%%%%%%%%%%%%%%%%%%%%%%%%%%%%%%%%%%%%%%%%%
%%%%%%%%%%%%%%%%%%%%%%%%%%%%%%%%%%%%%%%%%%%%%

For the \emph{weaker} end of the algebraic proof systems' hierarchy many size lower bounds are known. Beginning in the  works of Beame \textit{et al.}~\cite{BeameIKPP96} and Buss \textit{et al.}~\cite{BussIKPRS96} on Nullstellensatz, through the first Polynomial Calculus (PC) lower bound by Razborov \cite{Razb98}, and the PC subset-sum lower bound by  Impagliazzo, Pudlak and Sgall \cite{IPS99} (the simplest form of the subset sum principle, also called sometimes \textit{Knapsack}, is the unsatisfiable over 0-1 values equation  $\sum_{i=1}^n x_i -\beta=0$, for $\beta>n$), as well as many other results.

Only recently, lower bounds against \emph{stronger} algebraic proof systems were established. Forbes, Shpilka, Tzameret and Wigderson \cite{FSTW21} considered subsystems of IPS using read-once oblivious algebraic programs (roABP) and multilinear formulas over large fields. However these subsystems  are not necessarily comparable with constant-depth fragments of IPS which are the focus of the current work.

Alekseev \cite{Ale21} established  lower bounds  against the Polynomial Calculus with additional extension variables (i.e., variables that abbreviate polynomials with a single fresh variable) over large fields.  This result is quite strong, since the proof system simulates strong propositional-logic systems like extended Frege\footnote{Though the hard instance is not a CNF, and so  the lower bound  in \cite{Ale21} does not imply Extended Frege lower bounds.}. However, this proof system  is (apparently) weaker than IPS. Furthermore, the complexity of proofs in this system is measured by bit-size (i.e., coefficients of monomials in each proof-line are written using binary notation; hence, even a short proof using a small number of steps can incur an exponential blow-up if it uses coefficients of super-exponential magnitude, as shown by Alekseev). Lastly, the hard instance in \cite{Ale21} uses coefficients of exponential magnitude, and this is crucial to the lower bound argument.

Very recently, Impagliazzo, Mouli and Pitassi  \cite{IMP22} established  lower bounds against PC with restricted number of extension variables over finite fields for CNF formulas. However this proof system is apparently weaker (or incomparable to)   constant-depth IPS, and is weaker than PC with proof-lines written as constant-depth circuits, because of the restriction on the number of allowed extension variables.
\medskip

The following lower bounds form the frontiers of what is known about the complexity of strong algebraic proof systems  most relevant to our work (i.e., IPS of increasing depth, beginning from Nullstellensatz, which is equivalent to depth-2 IPS,  up to unbounded depth IPS):

\begin{enumerate}
\item[(i)] \emph{Conditional} lower bounds against IPS proofs for the Binary Value Principle $\sum_{i=1}^n 2^{i-1}x_i = -1$ (a subset sum instance with coefficients of exponential magnitude) by Alekseev, Hirsch, Grigoriev and Tzameret \cite{AGHT20}. Apart from this result being conditional, the hard instances use coefficients of exponential magnitude, and this is crucial to the lower bound argument.

\item[(ii)] Andrews and Forbes \cite{AF22} very recently proved constant-depth IPS lower bounds. However, the hard instance itself cannot be computed by a polynomial-size constant-depth circuit, and this fact is crucial to the lower bound proof.
\end{enumerate}

%%%%%%%%%%%%%%%%%%%%%%%%%%%%%%%%%%%%%%%%%%%%%
%%%%%%%%%%%%%%%%%%%%%%%%%%%%%%%%%%%%%%%%%%%%%
\subsection{Our Results}
%%%%%%%%%%%%%%%%%%%%%%%%%%%%%%%%%%%%%%%%%%%%%
%%%%%%%%%%%%%%%%%%%%%%%%%%%%%%%%%%%%%%%%%%%%%

We establish super-polynomial constant-depth IPS lower bounds for a subset sum instance with small coefficients (i.e., 0-1 coefficients) that is computable by an $O(n^5)$-size \emph{depth-2} circuits, and where the IPS proof is multilinear. 

To understand better the proof system we work against, recall the proof shown in \autoref{eq:NS-first}, in which the $g_i$'s are written as algebraic circuits---this is (equivalent) to the general IPS system. 
We shall work with the proof system multilinear \lbIPS, following the notation \lbIPS\ in \cite{FSTW21}. Proofs in  
multilinear \lbIPS\ of the unsatisfiability of $\{f_i(\vx)=0\}$ are (roughly; the  actual proof system is in fact stronger than this, see \autoref{def:IPS}) defined as  the following polynomial identity 
$$
 \sum_i g_i(\vx)\cd f_i(\vx) + \sum_jh_j(\vx)\cd (x_j^2-x_j)=1
$$
where $h_j(\vx)$ are some polynomials and the $ g_i(\vx)$'s are \emph{multilinear} polynomials, and the $g_i(\vx)$'s and $h_j(\vx)$'s are all written with constant-depth circuits (but \emph{not} necessarily multilinear formulas,  in contrast to multilinear-formula \lbIPS\ as in \cite{FSTW21}).

\begin{theorem}[Informal; see \autoref{thm:main}]\label{thm:ovr:main}
Every constant-depth multilinear \lbIPS\ refutation of the subset sum variant $\hardzp$ (for $\beta\not\in\{0,\dots,n^4\}$) requires super-polynomial (in $n$) size.
\end{theorem}

% To understand better the proof system we work against, recall the proof shown in \autoref{eq:NS-first}, in which the $g_i$'s are written as algebraic circuits---this is (equivalent) to the general IPS system. 
% We shall work with the proof system multilinear \lbIPS, following the notation \lbIPS\ in \cite{FSTW21}. Proofs in  
% multilinear \lbIPS\ are (roughly; the  actual proof system is in fact stronger than this, see XXX)  written as $$
%  \sum_i g_i(\vx)\cd y_i + h(\vx,\vz),
% $$
% where $h(\vx,\vz)$ is a polynomial in $\vx,\vz$ such that $h(\vx,\vnz)=0$, and the $ g_i(\vx)$'s are multilinear polynomials, and the $g_i(\vx)$'s and $h(\vx,\vz)$ are written with constant-depth circuits (but \emph{not} necessarily multilinear formulas,  in contrast to multilinear-formula \lbIPS\ as in \cite{FSTW21}), and such that following is a polynomial identity: $$
%  \sum_i g_i(\vx)\cd f_i(\vx) + h(\vx,x_1^2-x_1,\dots,x_n^2-x_n) = 1\,.
% $$ 
% 

%%%%%%%%%%%%%%%%%%%%%%%%%%%%%%%%%%%%%%%%%%%%%
%%%%%%%%%%%%%%%%%%%%%%%%%%%%%%%%%%%%%%%%%%%%%
\para{Significance of the Results and Context.}
%%%%%%%%%%%%%%%%%%%%%%%%%%%%%%%%%%%%%%%%%%%%%
%%%%%%%%%%%%%%%%%%%%%%%%%%%%%%%%%%%%%%%%%%%%%

This is the first constant-depth IPS lower bound on an instance that is computable itself with small constant-depth circuits, when the polynomial that constitutes the IPS proof is multilinear (see \autoref{def:IPS}).
Our hard instance  has coefficients of small magnitude, and the lower bound is in the stronger unit cost model of  algebraic circuits (i.e., in terms of the size of the circuits, not the size of the binary representation of the coefficients appearing  in them).
Thus, we rectify all the purported shortcomings of previous constant-depth algebraic proofs lower bounds (while paying by requiring that the IPS proofs are (partially) multilinear).

\autoref{thm:ovr:main} contributes to the  tradition of showing that simple subset sum variants are hard for algebraic proofs. While Impagliazzo \textit{et al.}~\cite{IPS99} initially showed that the subset sum principle requires exponentially many monomials in PC refutations, and Forbes \textit{et al.}~\cite{FSTW21} extended this to roABP and multilinear formulas, we show this hardness holds at least up to constant-depth IPS (when the proofs are multilinear).

Subset sum variants are not  translations of CNF formulas or Boolean formulas more generally. Hence, algebraic-proofs lower bounds for them do not imply (immediately at least) propositional logic proof size lower bounds (i.e., Frege-style proofs). However, a major motivation behind investigating the complexity of algebraic proof systems is to understand the power of algebraic reasoning and proofs (and their algorithmic counterpart, e.g., Gr\"obner basis computations).
% in addition to shedding light and eventually to resolve major open problems in the (Boolean) proof complexity of propositional proof systems operating with logical gates and counting gates modulo a prime (mod p gates, namely $\ACZ[p]$-Frege proofs).
For this purpose, it is enough to prove lower bounds on  hard instances that are not necessarily translations of CNFs or Boolean formulas. Indeed many works on the complexity of algebraic proof systems are dedicated to establishing such lower bounds, most prominently the subset sum principle, and its variants (see also Razborov \cite{Razb98} non-CNF pigeonhole principle).

Furthermore, lower bounds on the size of algebraic proofs of subset-sum instances, and generally instances  from the  language of unsatisfiable 0-1 multivariate polynomials over a field, are as relevant to  the Cook's programme mentioned above as much as Boolean formulas. The reason is that this language is a \coNP-complete language, since we can efficiently check if a given 0-1 assignment satisfies all the polynomials in the system (assuming the polynomials and field elements are written   in some standard way), and Boolean unsatisfiability is easily reducible to this language.
% Thus,  lower bounds on algebraic proofs (even for non-CNF instances) is a legitimate part of Cook's programme, mentioned above.

We explain in what follows our main technical contribution, which can be of independent interest.

%%%%%%%%%%%%%%%%%%%%%%%%%%%%%%%%%%%%%%%%%%%%%
%%%%%%%%%%%%%%%%%%%%%%%%%%%%%%%%%%%%%%%%%%%%%
\subsection{Proof Technique}
%%%%%%%%%%%%%%%%%%%%%%%%%%%%%%%%%%%%%%%%%%%%%
%%%%%%%%%%%%%%%%%%%%%%%%%%%%%%%%%%%%%%%%%%%%%

Our proof draws techniques from two sources. We use the methods presented in Limaye, Srinivasan and Tavenas  \cite{LST21} to prove superpolynomial lower bounds for constant-depth algebraic circuits, and combine these with the functional lower bound framework of Forbes, Shpilka, Tzameret and Wigderson \cite{FSTW21} for size lower bounds on IPS proofs (see also \cite{FKS16} for the functional lower bound approach in algebraic circuit complexity in general).

In general, we prove \autoref{thm:ovr:main} by  reducing the task of lower bounding the size of a constant-depth algebraic circuit  computing the multilinear polynomial that constitutes the IPS proof  into the following task: lower bound the size of a  constant-depth \emph{set-multilinear} \emph{circuit} computing an  associated \emph{set-multilinear polynomial}. To get the new associated set-multilinear polynomial from the original multilinear IPS proof (which is not necessarily set-multilinear by itself) we use a variant of the functional lower bound approach with some additional arguments that we introduce to deal with the need to focus on set-multilinear monomials  within a polynomial that is not set-multilinear. 

Once, we have the associated set-multilinear polynomial we can use the reduction presented in \cite{LST21} from constant-depth general circuits to constant-depth set-multilinear circuits. The reduction loses only a constant-factor in the depth, but pays quite heavily in the degree of the set-multilinear output polynomial. So in order to keep the size of the obtained set-multilinear circuit reasonable, we need to restrict the degree of the set-multilinear polynomial considerably.

%\textcolor[rgb]{0.690196,0.690196,0.690196}{Tuomas: we have an f that is the refutation-modulo of the z'ss (the f from theorem 1). By assigning some of the z we get both the KS-w and accordingly a  refutaion modulo  of the KS-w axioms. Now we use the functional lower bound approach of FSTW21. }

Notice that unlike in Limaye \textit{et al.}~\cite{LST21} (or in circuit complexity in general), we do not work from the get-go with a set-multilinear polynomial for which we need to prove a lower bound against constant-depth circuits computing it. In our case, we need to somehow show that this set-multilinear polynomial is ``embedded'' in some way in \emph{any} multilinear IPS proof of our hard instance. This is the main technical challenge we face in this work.

Another point is that to show our simple  degree-2 instance from  \autoref{thm:ovr:main} is hard we use a substitution in this simple instance. Specifically, by assigning some of the $z_{ijk\ell}$ and $x_i$ variables in the hard instance, we show that one gets another variant of subset sum denoted $\ks_w$ ($\ks$ stands for Knapsack) that is defined with respect to some word $w\in\ZZ^d$ (in the sense of \cite{LST21}). The definition of the subset sum over $w$ is designed so that the multilinear IPS refutation  of the simple hard instance from Theorem 1 ``embeds''  $P_w$ after applying the substitution to the $z_{ijk\ell}$ (and $x_i$) variables, where $P_w$ is the word polynomial from \cite{LST21}, which induces a full-rank coefficient matrix on a set of all set-multilinear monomials that arise from the given word $w$ (a coefficient matrix of a polynomial is an associated matrix whose rank serves as a complexity measure for the polynomial's circuit size, and in which each entry is a coefficient of a specific monomial in the polynomial). 

%\textcolor{red}{is the homogeneous polynomial containing all possible set-multilinear monomials with respect to  a variable-partition (based on $w$) from \cite{LST21}, meaning this $P_w$  induces a full-rank coefficient matrix (with respect to $w$}%

The meaning of a polynomial ``embedding'' a set-multilinear polynomial refers to the set-multilinear polynomial being the set-multilinear \emph{projection} of the original polynomial. Hence, we need to consider the \emph{projection} to the space of all set-multilinear polynomials over a particular variable-partition of our original polynomial. We extend the evaluation dimension method from \cite{FSTW21} to prove a rank lower bound for the coefficient matrix of this set-multilinear projection, which yields the set-multilinear circuit lower bound via a lemma from \cite{LST21}.

%\textcolor[rgb]{0.815686,0.815686,0.815686}{The set-multilinear %polynomial we reduce to arises from a suitable substitution instance of $f$. An immediate obstacle here is the fact that we do not have an explicit representation of $f$. However the \textcolor{green}{substitution????} instance is itself a refutation of some substitution instance of the polynomial $\hardzp$, and  the idea is to use techniques from \cite{FSTW21} to prove a lower bound for this refutation. The substitution instance of $f$  is however not set-multilinear itself and thus we need to consider its \emph{projection} to the space of all set-multilinear polynomials over a particular variable-partition. We extend the evaluation dimension method from \cite{FSTW21} (cf.~\textcolor{green}{xxx}) to prove a rank lower bound for the coefficient matrix of this set-multilinear projection, which yields the set-multilinear circuit lower bound via a lemma from \cite{LST21}. }

\medskip

To add on the above, our proof diverges from its forbears in some essential ways. Firstly, as mentioned before, when \cite{LST21} can work from the get-go with a low-degree set-multilinear polynomial, the multilinear refutations we consider are not of low-degree nor are they set-multilinear. Thus we need to find suitable set-multilinear polynomials within the refutations, and consider projections to the space of set-multilinear polynomials with respect to some variable-partition. Secondly, we use the method based on partial assignments (or evaluations) from \cite{FSTW21} to prove our rank lower bound. Our use of these partial assignments is however more subtle than the evaluation dimension method of \cite{FSTW21}. 

Forbes \textit{\textit{et al.}}~\cite{FSTW21} showed a rank lower bound against a coefficient matrix of a polynomial by reducing it to \emph{dimension} lower bound for the space of all multilinearizations of the polynomials obtained after  appropriate partial assignments (to the $\vy$-variables of a polynomial in both $\vx$ and $\vy$ variables; see \cite{FSTW21}). The dimension lower bound is then proved by showing that enough linearly independent  monomials appear in the space as leading monomials. Such an argument is not enough in our case, because we need to prove the existence of specific set-multilinear monomials (with respect to to a words $w$) in the original polynomial (namely, the IPS proof polynomial). The root of this difference is that the original argument defines partial assignments over \emph{all} the monomials in the hard polynomial,  while we try to argue for a rank lower bound for the coefficient matrix of only a \emph{part }of the monomials, namely only the  set-multilinear monomials within a non-set-multilinear polynomial---and this corresponds to only a \emph{submatrix} of the full coefficient matrix. The drawback here is that this also forces us to consider only multilinear refutations (because multilinearization as used in \cite{FSTW21} does not increase rank but preserves the high rank of the coefficient matrix of the original polynomial [before multilinearization], while multilinearization of  a polynomial \emph{can} increase the rank of the set-multilinear coefficient \emph{sub}-matrix; e.g., multilinearizing the \mbox{non-set-multilinear} monomial $x^2_1y_2^3$ produces $x_1 y_2$ which is set-multilinear, assuming the variables are partitioned into $\vx$ and $\vy$ variables).

%\textcolor[rgb]{0.752941,0.752941,0.752941}{The FSTW arguments do show a rank lower bound also for non-multilinear polynomials. The multilinearization occurs in the reduction from rank to evaluation dimension lower bound. We cannot do this as then the columns and rows outside our submatrix interfere with the argument.}

%  -  \textcolor[rgb]{0.690196,0.690196,0.690196}{We on the other hand argue that the leading monomial of the coefficient of any sml monomial in y variables is the corresponding sml monomial in x-variables.} \textcolor[rgb]{0.752941,0.752941,0.752941}{So it is actually even more specific than what I've wrote. And the reason why we need to get so nitty gritty is the fact that we consider a submatrix of the coefficient matrix.}

Note also that our proof technique is completely different from Andrews and Forbes \cite{AF22} who used the \emph{hard multiples} framework from \cite{FSTW21}.

\section{Preliminaries}
%%%%%%%%%%%%%%%%%%%%%%%%%%%%%%%%%%%%%%%%%%%%%
%%%%%%%%%%%%%%%%%%%%%%%%%%%%%%%%%%%%%%%%%%%%%

\subsection{ Polynomials and Algebraic Circuits}\label{sec:algebraic_circuits}

For excellent treatises on algebraic circuits and their complexity see Shpilka and Yehudayoff \cite{SY10} as well as Saptharishi \cite{Sap17-survey}.
 Let \G\ be a ring. Denote by $\G[X]$ the ring of (commutative) polynomials with coefficients from $ \G $ and variables $X:=\{x_1,x_2,\,\dots\,\}$. A \emph{polynomial }is a formal linear combination of monomials, where a \emph{monomial} is a product of variables. Two polynomials are \emph{identical }if all their monomials have the same coefficients. 

The (total) degree of a monomial is the sum of all the powers of variables in it. The (total) \emph{degree} of a polynomial is the maximal total degree of a monomial in it. The degree of an \emph{individual} variable in a monomial is its power. The \emph{individual degree} of a monomial is the maximal individual degree of  its variables. The individual degree of a polynomial is the maximal individual degree of its monomials. For a polynomial $f$ in $\G[X,Y]$ with $X,Y$ being pairwise disjoint sets of variables,  the \emph{individual $Y$-degree} of $f$ is the maximal individual degree of a $Y$-variable only in $f$.

Algebraic circuits and formulas over the ring \G\   compute polynomials in $\G[X]$ via addition and multiplication gates, starting from the input variables and constants from the ring. More precisely, an \emph{algebraic circuit} $C$ is a finite directed acyclic graph (DAG) with \textit{input nodes} (i.e., nodes  of in-degree zero) and a single \textit{output node} (i.e.,  a node of out-degree zero).  Edges are labelled by ring \G\ elements.  Input nodes are labelled with variables or scalars from the underlying ring. In this work (since we work with constant-depth circuits) all  other nodes have unbounded \emph{fan-in} (that is, unbounded in-degree) and are labelled by either an addition gate $+$ or a product gate $\times$.
Every node in an algebraic circuit $C$ \emph{computes} a polynomial in $\G[X]$ as follows: an input node computes  %
% and that moreover the two edges going into a gate $v$   labeled with  $\times$ % or $+$ are
%  labeled with  \emph{left} and \emph{right}.  This
%is to determine the order of addition and multiplication\footnote{Although ultimately, %addition and multiplication are commutative. }.
%
the variable or scalar that   labels  it. A $+$ gate
 computes the linear combination of all the polynomials computed by its incoming nodes, where the coefficients of the linear combination are determined by the corresponding incoming edge labels. A $\times$ gate computes the product of all the polynomials computed by its incoming nodes (so edge labels in this case are not needed). The polynomial computed by a node $u$ in an algebraic circuit $C$ is denoted $\widehat u$. Given a circuit $C$, we denote by $\widehat C$ the polynomial computed by $C$, that is, the polynomial computed by the output node of $C$.  The \emph{\textbf{size}} of a circuit $C$ is the number of nodes in it, denoted $|C|$, and the \emph{\textbf{depth}} of a circuit is the length of the longest directed path in it (from an input node to the output node). The \textbf{\emph{product-depth }}of the circuit is the maximal number of product gates in a directed path from  an input node to the output node.

% an algebraic circuit $C$ we write $C(a/x)$ to denote the \emph{substitution instance} of $C$ in which every occurrence of the node $x$ is replaced by the sub-circuit $a$; in case  $C(x)$ is written with its displayed variable(s) $x$ we can write $C(x)(a/x)$ for this substitution instance.

We say that a polynomial is \emph{homogeneous} whenever every monomial in it has the same (total) degree. We say that a polynomial is \emph{multilinear} whenever the individual degrees of each of  its variables are at most 1. 

Let $\overline{X} = \langle X_1,\ldots,X_d\rangle$ be a sequence of pairwise disjoint sets of variables, called \emph{variable-partition}. We call a monomial $m$ in the variables $\bigcup_{i\in [d]}X_i$  \emph{set-multilinear} over the variable-partition $\overline{X}$ if it contains exactly one variable from each of the sets $X_i$, i.e. if there are $x_i\in X_i$  for all $i\in [d]$ such that $m = \prod_{i\in [d]}x_i$. A polynomial $f$ is set-multilinear over $\overline{X}$ if it is a linear combination of set-multilinear monomials over $\overline{X}$. For a sequence $\overline{X}$  of sets of variables, we denote by $\FF_{\sml}[\overline{X}]$ the space of all polynomials that are set-multilinear over $\overline{X}$.

We say that an algebraic circuit $C$ is set-multilinear over $\overline{X}$ if $C$ computes a polynomial that is set-multilinear over $\overline{X}$, and each internal node of $C$  computes a polynomial that is set-multilinear over some sub-sequence of $\overline{X}$.

\subsection{Strong Algebraic Proof Systems}
For a survey about algebraic proof systems and their relations to algebraic complexity see the survey \cite{PT16}.
Grochow and Pitassi~\cite{GP14}  suggested the following algebraic proof system  which is essentially a Nullstellensatz proof system (\cite{BeameIKPP96}) written as an algebraic circuit. A proof in the  Ideal Proof System is given as  a \emph{single} polynomial. We provide below the \emph{Boolean} version of  IPS (which includes the Boolean axioms), namely the version that establishes the unsatisfiability over 0-1 of a set of polynomial equations.  In what follows we follow the notation in \cite{FSTW21}:

%\iddo{shall we fix the field from now on to \reals?}\hirsch{No. %First of all, should we first formulate it without %$x^2-x$? (In this case we would need an algebraically %closed field.) Second, we will be talking about $\mathbb{Q}$ %at least.}

% This is potentially much more powerful as there are polynomials such as the determinant which are of high degree and involve exponentially many monomials and yet can be computed by small algebraic circuits. They named the resulting system the \emph{Ideal Proof System (IPS)} which we now define.

\begin{definition}[Ideal Proof System (IPS),
Grochow-Pitassi~\cite{GP14}]\label{def:IPS} Let $f_1(\vx),\ldots,f_m(\vx),p(\vx)$ be a collection of polynomials in $\F[x_1,\ldots,x_n]$ over the field \F. An \demph{IPS proof of $p(\vx)=0$ from axioms $\{f_j(\vx)=0\}_{j=1}^m$}, showing that $p(\vx)=0$ is semantically  implied from the assumptions $\{f_j(\vx)=0\}_{j=1}^m$ over $0$-$1$ assignments, is an algebraic circuit $C(\vx,\vy,\vz)\in\F[\vx,y_1,\ldots,y_m,z_1,\ldots,z_n]$ such that (the equalities in what follows stand for  formal polynomial identities\footnote{That is, $C(\vx,\vnz,\vnz)$ computes the zero polynomial and $C(\vx,f_1(\vx),\ldots,f_m(\vx),x_1^2-x_1,\ldots,x_n^2-x_n)$ computes the polynomial $p(\vx)$.}):
        \begin{enumerate}
                \item $C(\vx,\vnz,\vnz) = 0$; and\vspace{-5pt}
                \item $C(\vx,f_1(\vx),\ldots,f_m(\vx),x_1^2-x_1,\ldots,x_n^2-x_n)=p(\vx)$.
        \end{enumerate}
        The \demph{size of the IPS proof} is the size of the circuit $C$. An \IPS\ proof  $C(\vx,\vy,\vz)$ of  $1=0$ from $\{f_j(\vx)=0\}_{j\in[m]}$ is called an \demph{IPS refutation} of $\{f_j(\vx)=0\}_{j\in[m]}$ (note that in this case  it must hold that  $\{f_j(\vx)=0\}_{j=1}^m$ have no common solutions in $\bits^n$).
If $\widehat C$ (the polynomial computed by $C$) is of individual degree $\le 1$ in each $y_j$ and $z_i$, then this is a \demph{linear} IPS refutation (called \emph{Hilbert} IPS by Grochow-Pitassi~\cite{GP18}), which we will abbreviate as \lIPS. If $\widehat C$ is of individual degree $\le 1$ only in the $y_j$'s then we say this is an \lbIPS refutation (following \cite{FSTW21}). If  $\widehat C(\vx,\vy,\vnz)$ is of individual degree $\le 1$ in each $x_j$ and $y_i$ variables, while $\widehat C(\vx,\vnz,\vz)$ is not necessarily multilinear, then this is a \demph{multilinear} \lbIPS\ refutation. %\mar{italic no lips'}  

If $C$ is of depth at most $d$, then this is  called a depth-$d$ \IPS\ refutation, and further called a depth-$d$ \lIPS refutation if $\widehat C$ is linear in $\vy,\vz$, and a depth-$d$ \lbIPS refutation if $\widehat C$ is linear in $\vy$, and depth-$d$ multilinear \lbIPS\ refutation if $\widehat C(\vx,\vy,\vnz)$ is linear in $\vx,\vy$. 
\end{definition}

Notice that the definition above adds the  equations $\{x_i^2-x_i=0\}_{i=1}^n$, called the  \demph{Boolean axioms} denoted  $\vx^2-\vx$, to the system $\{f_j(\vx)=0\}_{j=1}^m$. This allows  to refute over $\bits^n$ unsatisfiable systems of equations. The variables $\vy,\vz$ are  called the \emph{placeholder} \emph{variables} since they are used as placeholders for the axioms. Also, note that the first equality in the definition of IPS means that the polynomial computed by $C$ is in the ideal generated by $\overline y,\overline z$, which in turn, following the second equality, means that $C$ witnesses the fact that $1$ is in the ideal generated by $f_1(\vx),\ldots,f_m(\vx),x_1^2-x_1,\ldots,x_n^2-x_n$ (the existence of this witness, for unsatisfiable set of polynomials, stems from the Nullstellensatz theorem \cite{BeameIKPP96}).    
\medskip 

In this work we focus on multilinear \lbIPS\ refutations. This proof system is complete because its \emph{weaker} subsystem multilinear-formula \lbIPS\  was shown in \cite[Corollary 4.12]{FSTW21} to be complete (and to simulate Nullstellensatz with respect to sparsity by already depth-2 multilinear \lbIPS\ proofs). 

To build an intuition for multilinear \lbIPS\ it is useful to consider a subsystem of it in which refutations are written as 
$$
C(\vx,\vy,\vz)= \sum_i g_i(\vx)\cd y_i + C'(\vx,\vz),$$ where $\widehat C'(\vx,\vnz)=0$ and the $\widehat g_i$'s are multilinear. Note indeed that $C(\vx,\vnz,\vnz)=0$ so that the first condition of IPS proofs holds, and that $C(\vx,\vy,\vnz)$ is indeed multilinear in $\vx,\vy$.

%In fact, \emph{we are going to assume a stronger version of multilinear} \lbIPS\ in which a proof $C(\vx,\vy,\vz)$ is such that $\widehat C(\vx,\vy,\vnz) $ is multilinear (in both $\vx$ and $\vy$), while $\widehat C(\vx,0,\vz) $ is \emph{not necessarily} multilinear. 

\medskip 

\textbf{Important remark}: Unlike the multilinear-formula \lbIPS\ in \cite{FSTW21}, in  multilinear \lbIPS\ refutations $C(\vx,\vy,\vz)$ we do \emph{not} require that the refutations are written as multilinear \emph{formulas} or multilinear \emph{circuits}, only that the \emph{polynomial} \emph{computed }by $C(\vx,\vy,\vnz)$ is multilinear, hence the latter proof system easily simulates the former. 

%--- Note that FSTW proved lower bounds via the function lower bounds approach also for IPS-LIN, and we only reqyuire also the polynomial to be multilinear.

% %%%%%%%%%%%%%%%%%%%%%%%%%%%%%%%%%%%%%%%%%%%%%
% %%%%%%%%%%%%%%%%%%%%%%%%%%%%%%%%%%%%%%%%%%%%%
% \subsection{F\textcolor{green}{urther Notatio}n}
% %%%%%%%%%%%%%%%%%%%%%%%%%%%%%%%%%%%%%%%%%%%%%
% %%%%%%%%%%%%%%%%%%%%%%%%%%%%%%%%%%%%%%%%%%%%%
% 
% ...

\subsection{Set-Multilinear Monomials over a Word}

We recall some notation from \cite{LST21}. Let $w\in\ZZ^d$  be a word. For a subset $S\subseteq  [d]$ denote by $w_S$  the sum $\sum_{i\in S} w_i$, and by $w|_S$  the \textbf{subword} of $w$  indexed by the set $S$. Let\footnote{The $P_w$ here is not to be confused with the canonical full-rank set-multilinear polynomial in \cite{LST21} denoted as well by $P_w$ mentioned in the introduction.}
\[
P_w := \{i\in [d] : w_i\geq 0\}
\]
be the set of \textbf{positive indices} of $w$ and let
\[
N_w := \{i\in [d] : w_i < 0 \}
\]
be the set of \textbf{negative indices} of $w$.

Given a word $w$, we associate with it a sequence $\overline{X}(w) = \langle X(w_1),\ldots,X(w_d)\rangle$ of sets of variables, where for each $i\in [d]$  the size of $X(w_i)$ is $2^{|w_i|}$. We call a monomial set-multilinear over a word $w$ if it is set-multilinear over the sequence $\overline{X}(w)$.

%?
%
%?
%
%\textcolor{green}{Tumoas to define here setmultilinear monomials and polynmalisl  circuits!!!!!!!!r
%}
  For a word $w$, let $\Pi_w$ denote the projection onto the space $\FF_\sml[\overline{X}(w)]$ that maps the set-multilinear monomials over $w$  identically to themselves and all other monomials to $0$.

\subsection{Relative Rank}

Let $M^P_w$  and $M^N_w$  denote the set-multilinear monomials over $w|_{P_w}$ and $w|_{N_w}$, respectively. Let $f\in\FF_{\sml}[\overline{X}(w)]$ and denote by $M_w(f)$   the matrix with rows indexed by $M^P_w$ and columns indexed by $M^N_w$, whose $(m,m')$th entry is the coefficient of the monomial $mm'$ in $f$.

For any $f\in\FF_{\sml}[\overline{X}(w)]$ define the \textbf{relative rank} with respect to $w$ as follows
\[
\relrk_w(f) = \frac{\rank(M_w(f))}{\sqrt{|M^P_w|\cdot|M^N_w|}}.
\]

\subsection{Monomial Orders}

Finally we recall some basic notions related to monomial orders. For an in-depth introduction see \cite{CoxLittleOShea07}. A monomial order (in a polynomial ring $\FF[X]$) is a well-order $\leq$ on the set of all monomials that respects multiplication:
\[
\text{if }m_1\leq m_2\text{, then } m_1m_3\leq m_2m_3\text{ for any }m_3.
\]
It is not hard to see that any monomial order extends the submonomial relation: if $m_1m_2 = m_3$  for some monomials $m_1,m_2$ and $m_3$, then $m_1\leq m_3$. This is essentially the only property we need of monomial orderings, and thus our results work for any monomial ordering. Given a polynomial $f\in\FF[X]$, the leading monomial of $f$, denoted $\LM(f)$, is the highest monomial with respect to $\leq$ that appears in $f$  with a non-zero coefficient.

%%%%%%%%%%%%%%%%%%%%%%%%%%%%%%%%%%%%%%%%%%%%%
%%%%%%%%%%%%%%%%%%%%%%%%%%%%%%%%%%%%%%%%%%%%%

%%%%%%%%%%%%%%%%%%%%%%%%%%%%%%%%%%%%%%%%%%%%%
%%%%%%%%%%%%%%%%%%%%%%%%%%%%%%%%%%%%%%%%%%%%%
\section{The Lower Bound}\label{sec:lower bound}
%%%%%%%%%%%%%%%%%%%%%%%%%%%%%%%%%%%%%%%%%%%%%
%%%%%%%%%%%%%%%%%%%%%%%%%%%%%%%%%%%%%%%%%%%%%

Our main theorem is as follows:

\begin{theorem}[Main]\label{thm:main}
Let $n,\Delta\in\NN_+$ with $\Delta\leq \nicefrac{1}{4}\log\log\log n$, and assume that $\chara(\FF) = 0$. Then any product-depth at most $\Delta$ multilinear \lbIPS\ refutation of the subset sum variant $\hardzp$ (for $\beta\not\in\{0,\dots,n^4\}$) is of size at least $ n^{(\log n)^{\exp(-O(\Delta))}}$.
\end{theorem}
 
Note that, when $\Delta \geq \omega(\log\log\log n)$, the lower bound above becomes trivial. To prove \autoref{thm:main} we need the following theorem, which is  proved in the sequel:
\begin{theorem}\label{theorem:hardz-lower-bound}
        Let $n,\Delta\in\NN_+$ with $\Delta\leq \nicefrac{1}{4}\log\log\log n$, and assume that $\chara(\FF) = 0$. Let $f$ be the multilinear polynomial such that 
        \[
        f = \frac{1}{\hardzp}\text{\quad over Boolean assignments.}
        \]
        Then, any circuit of product-depth at most $\Delta$ computing $f$ has size at least
 
        \[
        n^{(\log n)^{\exp(-O(\Delta))}}\,.
        \]
\end{theorem}

\begin{proof}[Proof of \autoref{thm:main} from \autoref{theorem:hardz-lower-bound}]
Let $C(\vx,\vy,\vw)$ be a multilinear \lbIPS\  refutation of $\hardz$ (where here for notational clarity we renamed the Boolean axioms placeholder variables from $\vz$ to $\vw$). Since there is only one non-Boolean axiom, $C$ has in fact only a single $\vy$ variable denoted $y$ (i.e., $\vy=\{y\}$). Note that  $\widehat C(\vx,y,\vnz)=g(\vx)\cd y$, for some polynomial $g(\vx)\in\F[\vx]$, because by assumption $\widehat C$ is linear in the $\vy$ variables (and since $\widehat C(\vx,0,\vnz)=0$). Therefore, $C(\vx,1,\vnz)$ computes the polynomial $g(\vx)$. Thus, the minimal product-depth-$\Delta$ circuit-size of $g(\vx)$ lower bounds the minimal product-depth-$\Delta$ circuit-size of $C(\vx,y,\vw)$. It remains to lower bound the size of product depth at most $\Delta$ circuits computing $g(\vx)$.  

Notice first that 
$$
\widehat C(\vx,y,\vw)= \widehat  C(\vx,y,\vnz) +\sum_{i} h_i\cd w_i,
$$
for some polynomials $h_i$ in $\vx,y,\vw$. Thus, $\widehat  C(\vx,y,\vw) =  g(\vx)\cd y + \sum_{i} h_i\cd w_i$. 

By definition of an IPS refutation 
$$\widehat C\left(\vx,\hardzp,\vx^2-\vx\right) = 1 
$$
and so 
$$g(\vx)\cd\left(\hardzp\right) + \sum_{i} (h_i\cd(x_i^2-x_i)) = 1.$$
From this we get that over the Boolean cube
$ g(\vx)\cd\hardzp \equiv 1$ (as a function, not necessarily as a polynomial identity), and hence $ g(\vx)  = \frac{1}{\hardzp}$ over the Boolean cube. This shows that the size of $C(\vx,y,\vw)$ must be at least $ n^{(\log n)^{\exp(-O(\Delta))}}$ by \autoref{theorem:hardz-lower-bound}.
\end{proof}

\bigskip 
The rest of \autoref{sec:lower bound} is devoted to proving  \autoref{theorem:hardz-lower-bound}.
\bigskip

We shall use the tight degree lower bound proved  in \cite{FSTW21} for functions defined by $\hat{f}(\vx)=\nicefrac{1}{f(\vx)}$ for simple polynomials $f(\vx)$. Specifically, we use the fact that any multilinear polynomial agreeing with $\nicefrac{1}{f(\vx)}$, where $f(\vx)$ is the subset-sum axiom $\sum_{i=1}^n x_i - \beta$ (where $\beta$ is such that the axiom has no Boolean roots) must have  degree $n$.  We note that a degree lower bound of $\ceil{\nicefrac{n}{2}}+1$ was established by Impagliazzo, Pudl\'ak, and Sgall~\cite{IPS99}.
%They actually established this degree bound\,
%
%%\footnote{The degree lower bound of Impagliazzo, Pudl\'ak, and Sgall~\cite{IPS99}  actually holds for the (dynamic) polynomial calculus proof system, while we only consider the (static) Nullstellensatz proof system here.  Note that for polynomial calculus Impagliazzo, Pudl\'ak, and Sgall~\cite{IPS99} also showed a matching upper bound of $\ceil{\nicefrac{n}{2}}+1$ for $\vaa=\vno$.}
%
%when $f(\vx)=\sum_i \alpha_i x_i-\beta$ for \emph{any} %$\vaa$, while we only consider $\vaa=\vno$ here. %
%
However, similar to \cite{FSTW21}, we need the tight bound of $n$ here as it will be used crucially in the proof of \autoref{lemma:rank-lower-bound} to obtain the rank lower bound, which is a stronger notion than degree lower bound. Recall the notation  $\vx^2-\vx$ for the Boolean axioms of the $\vx$-variables.

%\mar{FIX: inconsistency between $x$ as a sequence of variables and $\vx$}

\begin{lemma}[Proposition 5.3 in \cite{FSTW21}]\label{lem:FSTW21-deg-lb}
Let $n\ge 1$ and $\F$ be a field with $\chara(\F)>n$ (or $\chara(\F)=0$). Suppose that $\beta\in \F\setminus\{0,\ldots,n\}$. Let $f\in\F[x_1,\ldots,x_n]$ be a multilinear polynomial such that
        \[
                f(\vx)\left(\sum_i x_i-\beta\right)=1 \mod \vx^2-\vx
                \;.
        \]
        Then $\deg f=n$.
\end{lemma}

\subsection{Subset-Sum Based on a Word $w$}\label{section:subset-sum-over-word}
In this section we define an auxiliary polynomial we use to prove \autoref{theorem:hardz-lower-bound}. It is a variant of the subset-sum that is defined from a given word $w$. Let $w\in \ZZ^d$  be arbitrary word, and consider the sequence $\overline{X}(w)=\langle X(w_1),\dots,X(w_d)\rangle $ of sets of variables. We fix now a useful representation of the variables in $\overline{X}(w)$.

For any $i\in P_w$, we write the variables of $X(w_i)$ in the form $x^{(i)}_\sigma$, where $\sigma$ is a binary string indexed by the set (formally, a binary string  \emph{indexed} by a set $A$ is a function from $A$ to $\{0,1\}$):   
\[
A_w^{(i)} := \left[\sum_{\substack{i'\in P_w\\i' < i }}w_{i'} + 1,\sum_{\substack{i'\in P_w\\i'\leq i}}w_{i'}\right].
\] 
Hence, the size of $A_w^{(i)}$ is precisely $w_i$ and binary strings on the interval $A_w^{(i)}$ (i.e., $\{0,1\}^{A_w^{(i)}}$), allows $2^{|A_w^{(i)}|}=2^{w_i}$ possible strings, each corresponds to a different variable in $X(w_i)$.

Similarly, for any $j\in N_w$, we  write the variables of $X(w_i)$ in the form $y^{(j)}_\sigma$, where $\sigma$ is a binary string indexed by the set 
\[
B_w^{(j)} := \left[\sum_{\substack{j'\in N_w\\ j' < j}}|w_{j'}| + 1,\sum_{\substack{j'\in N_w\\ j' \leq j}}|w_{j'}|\right].
\]
We call the variables in $x^{(i)}_\sigma$  the \emph{positive variables}, or simply $\vx$-variables, and the variables $y^{(j)}_\sigma$  the \emph{negative variables}, or simply $\vy$-variables. We write $A_w^S$  for the set $\bigcup_{i\in S}A_w^{(i)}$  for any $S\subseteq P_w$, and $B_w^T$  for the set $\bigcup_{j\in T}B_w^{(j)}$ for any $T\subseteq N_w$.

Each monomial that is set-multilinear on $w|_{S}$ for some $S\subseteq P_w$ corresponds to a binary string indexed by the set $A_w^S$, and any monomial that is set-multilinear on $w|_{T}$   for some $T\subseteq N_w$ corresponds to a binary string indexed by the set $B_w^T$. For any set-multilinear monomial $m$ on some $w|_{S}$ with $S\subseteq P_w$ we denote by $\sigma(m)$ the corresponding binary string indexed by $A_w^S$, and for any binary string $\sigma$ indexed by $A_w^S$   we denote by $m(\sigma)$ the monomial it defines, and similarly for strings and monomials on the negative variables.
Thus notice that for a (negative or positive) monomial $m$ we have  $m(\sigma(m))=m$. Moreover, if $m$ is a negative monomial and $S\subseteq P_w$, we write $m(\sigma(m)|_{A^S_w})$ to denote the \emph{positive} monomial determined by the string $\sigma(m)|_{A^S_w}$ which is a substring of $\sigma(m)$ restricted to $A^S_w$.

Therefore, every set-multilinear monomial on $w$ is of degree $d$ with each $\vx$-variable picked uniquely from the $X(w_i)$-variables, for  $i\in P_w$ the positive indices in $w$, and each $\vy$-variable is picked uniquely from the $X(w_j)$-variables,  for $j\in N_w$ the negative indices in $w$, and moreover the  monomial corresponds to a binary string of length $\sum_{i=1}^d |w_i|$.     

We call the word $w$   \textbf{balanced} if for every $i\in P_w$  there is some $j\in N_w$ such that $A_w^{(i)}\cap B_w^{(j)}\neq \emptyset$ and for every  $j\in N_w$ there is some $i\in P_w$ such that  $A_w^{(i)}\cap B_w^{(j)}\neq \emptyset$. This means that any positive variable as defined above has some overlap with some negative variable within the given indexing scheme, and vice versa. \autoref{fig:unbalanced} and \autoref{fig:balanced} give examples of unbalanced and balanced words, respectively.

Our notion of a balanced word is different from, but related to, the notion of a  $b$-unbiased word in \cite{LST21}: a word $w\in\ZZ^d$ is $b$-\emph{unbiased} if $|w_{[t]}|\leq b$ for every $t\leq d$. If a balanced word $w\in\ZZ^d$ has all its entries bounded by $b$ in absolute value, i.e., $|w_i|\leq b$ for every $i\in [d]$, then the sum of all entries of $w$ is bounded by $b$ in the sense that $|w_{[d]}|\leq b$. The notion of balance is more relaxed than that of unbiased in the sense that we do not need the property that all initial segments are also bounded by $b$. On the other hand the construction and proof in \cite{LST21} does not require a balanced word, when this property is essential for us in the proof of \autoref{lemma:rank-lower-bound}.

\begin{figure}\center 
\includegraphics[scale=0.8]{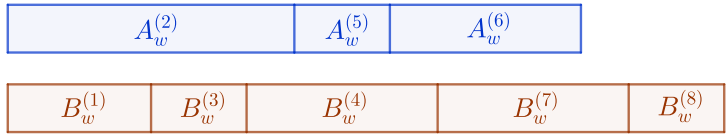}
\caption{\small An example of an \emph{un}balanced word $w\in\ZZ^d$, and where each $w_i$ is written as a box with its corresponding binary string index set inside it, such that  $|A_w^{(i)}|=w_i$ for $i\in P_w$ and $|B_w^{(j)}|=w_j$ for $j\in N_w$. A word is balanced if for every positive box  $i \in P_w$, there exists some overlapping negative box, i.e., a   $j \in N_w$ such that $A_w^{(i)} \cap B_w^{(j)} \neq \emptyset$, and vice versa. Here however, $B_w^{(8)} \cap (A_w^{(2)} \cup A_w^{(5)} \cup A_w^{(6)}) = \emptyset$.}
        \label{fig:unbalanced}
\end{figure}

\begin{figure}\center   
        \includegraphics[scale=0.8]{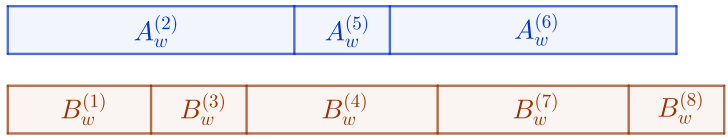}
        \caption{\small An example of a balanced word. The construction of the knapsack polynomial $\textbf{ks}_w$ makes use of a balanced word $w$.}
        \label{fig:balanced}
\end{figure}

We define the polynomial $\ks_w$ as follows. Below we suppose that $|w_{N_w}|\geq |w_{P_w}|$, so that the negative monomials are determined by a longer binary string than the positive ones. Otherwise we flip the roles of the negative and positive variables in the definition below. 

For a positive index $i\in P_w$ and $\sigma\in\{0,1\}^{A_w^{(i)}}$, define the polynomial 
\begin{equation}\label{lift}
 f^{(i)}_\sigma :=\prod_{\substack{j\in N_w:\\ A_w^{(i)}\cap B_w^{(j)} \neq \emptyset}}\sum_{\sigma_j\in \{0,1\}^{B_w^{(j)}}}y^{(j)}_{\sigma_j},
\end{equation}
where the sum ranges over all those $\sigma_j$  that agree with $\sigma$  on $A_w^{(i)}\cap B_w^{(j)}$. The degree of $f^{(i)}_\sigma$ equals the number of those $B_w^{(j)}$ that overlap with $A_w^{(i)}$. Note that the degree is always at least $1$. Now define the polynomial 
\[
 \ks_w:=\sum_{i\in P_w}\sum_{\sigma\in\{0,1\}^{A_w^{(i)}}} x_\sigma^{(i)}f_\sigma^{(i)} -\beta,
\]
where $\beta$ is any field \F\ element  so that the polynomial $\ks_w$ has no Boolean roots (here we use the fact that $\chara(\F)=0$). \autoref{fig:matching} gives an example of the construction.

\begin{figure} \center 
        \includegraphics[scale=0.8]{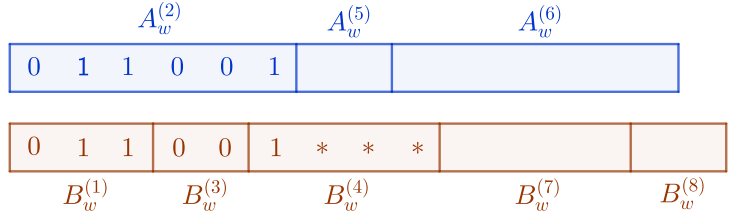}
        \caption{\small Here $\ast$ represents either $0$ or $1$. In the construction of the polynomial $\ks_w$, for $i=2$ and $\sigma = 011001$, we have $f_{011001}^{(2)} = y_{011}^{(1)} \cdot y_{00}^{(3)} \cdot (y_{1000}^{(4)} + y_{1001}^{(4)} + \cdots + y_{1111}^{(4)})$.}
        \label{fig:matching}
\end{figure}

The basic idea behind the above construction is simple. Given a monomial $m$ that is set-multilinear over $w|_{N_w}$, consider the partial assignment $\tau_m$ to the negative variables that sends any variable in $m$ to $1$  and all other negative variables to $0$. Now after applying $\tau_m$  to $\ks_w$ we are left with a simple subset sum instance
\[
\sum_{i\in P_w}x_{\sigma_i}^{(i)} - \beta,
\] 
where $\sigma_i$ is the binary string indexed by $A_w^{(i)}$  that agrees with $\sigma(m)$  on $A_w^{(i)}$. Similarly for any monomial that is set-multilinear over $w|_{T}$ for some $T\subseteq N_w$, the polynomial $\tau_m(\ks_w)$ is the subset sum instance
\[
\sum_{i\in S}x_{\sigma_i}^{(i)} - \beta,
\] 
where $S$  is the maximal subset of $P_w$  such that $A_w^S\subseteq B_w^T$. With \autoref{lem:FSTW21-deg-lb} we have a very good understanding of the multilinear IPS refutations of such subset sum instances, namely we know that the multilinear polynomial $f$  that equals
\[
\frac{1}{\sum_{i\in S}x_{\sigma_i}^{(i)} - \beta}\text{  over Boolean assignments}
\]
has as its leading monomials the product $\prod_{i\in S}x_{\sigma_i}^{(i)}$. This observation allows us to prove our rank lower bound in the following section. \autoref{fig:leading-monomial} in \autoref{sec:rank-lower-bound-lemma} illustrates the way assignments to the negative variables give rise to simple subset sum instances.

%where $S$ is the maximal subset of $P_w$ such that $A_w^S\subseteq B_w^T$. With \autoref{lem:FSTW21-deg-lb} we have a very good understanding of the refutations of such subset sum instances, namely we know exactly the leading monomials of the multilinear refutations. This observation allows us to prove our rank lower bound in the following section. \autoref{fig:matching} illustrates the way we can capture the leading monomials by assignments to the negative variables.
%\begin{figure}
%        \includegraphics{img/tex/matching.pdf}
%        \caption{Matching between $A_w^S$ and $B_w^T$}
%
%        \label{fig:matching}
%\end{figure}

% there are at most two indices $j\in I_N$ such that $A_i\cap B_j \neq \emptyset$. Let now $i\in I_P$ and let $\sigma\in\{0,1\}^{A_i}$. If there is only one $j$ such that $A_i\cap B_j \neq \emptyset$, let $f^{(i)}_\sigma$ be the sum $\sum_{\sigma'} y_{\sigma'}^{(j)}$, where $\sigma'$ ranges over the Boolean strings indexed by $B_j$   that agree with $\sigma$   on the index-set $A_i\cap B_j$. If there are two distinct $j$   and $j'$ with $A_i\cap B_j \neq \emptyset$ and $A_i\cap B_{j'} \neq \emptyset$, let
%$$
%f^{(i)}_\sigma:=\sum_{\sigma_j,\sigma_{j'}} y_{\sigma_j}^{(j)}y_{\sigma_{j'}}^{(j')}
%$$ where $\sigma_j$   ranges over the Boolean strings   indexed by $B_j$ that agree with $\sigma$   on the index-set $A_i\cap B_j$ and where $\sigma_{j'}$   ranges over the Boolean strings indexed by $B_{j'}$   that agree with $\sigma$   on the index set $A_i\cap B_{j'}$.

\subsection{Rank Lower Bound Lemma}\label{sec:rank-lower-bound-lemma}

In this section we prove the main technical lemma of this paper -- a rank lower bound for subset sum over any balanced word. Let $w\in\ZZ^d$  be any word, and let $f$ be a multilinear polynomial in the variables $\overline{X}(w)$. Denote by $M(f)$ the coefficient matrix of $f$  with rows indexed by all multilinear monomials (of \emph{any} degree) in the positive variables and columns indexed by all multilinear monomials (again, of any degree) in the negative variables, and denote by $M_w(f)$   its \emph{submatrix} with rows indexed by all monomials that are set-multilinear over $w|_{P_w}$  and columns indexed by all monomials that are set-multilinear over $w|_{N_w}$.

\begin{figure}\center 
        \includegraphics[scale=0.7]{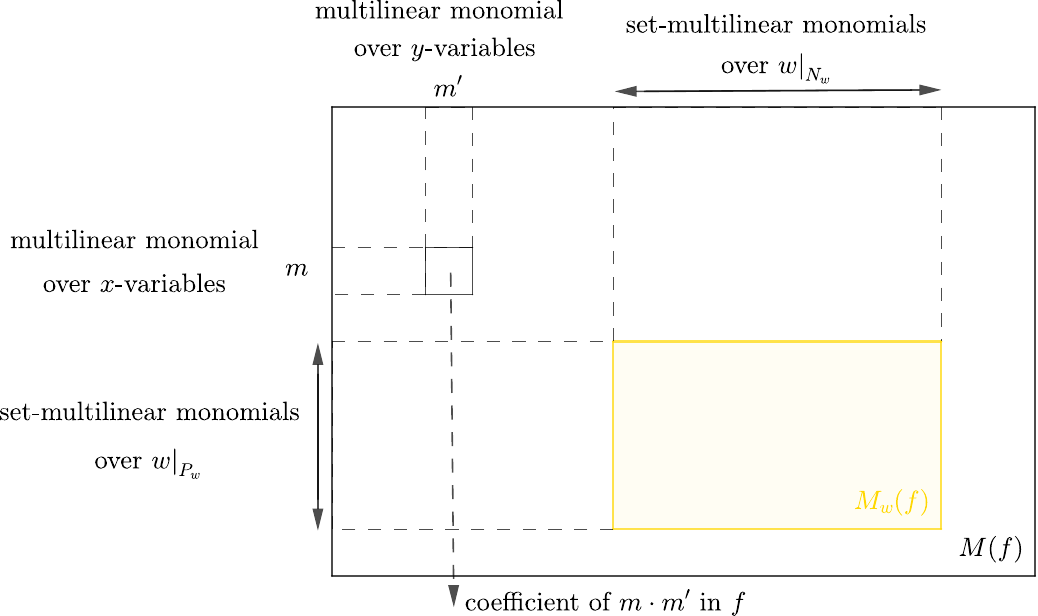}
        \caption{\small The coefficient matrix $M(f)$ with rows and columns indexed by the multilinear monomials over the $x$-variables and $y$-variables respectively; and its submatrix $M_w(f)$ with rows and columns indexed by the set-multilinear monomials over $w \vert_{P_w}$ and $w \vert_{N_w}$ respectively. If $w$ is a balanced word and $f$ is a multilinear polynomial agreeing with $\nicefrac{1}{\ks_w}$ over boolean assignments, then $M_w(f)$ has full-rank}
        \label{fig:submatrix}
\end{figure}

\begin{lemma}\label{lemma:rank-lower-bound}
                Let $w\in\ZZ^d$  be a balanced word, and let $f$ be the multi-linear polynomial  so that
                \[
                f = \frac{1}{\ks_w} \text{  over Boolean assignments}.
                \]
                Then $M_w(f)$ is full-rank.
        \end{lemma}

        \begin{proof}
                        Without loss of generality we assume that $|w_{N_w}| \geq |w_{P_w}|$, so that our notation matches that of Section \ref{section:subset-sum-over-word}.
                Write
\begin{equation}\label{eq:fgm}
                f = \sum_{m}g_m(x) m,        
\end{equation}              
where the sum ranges over all multilinear monomials $m$ in the $y$-variables and $g_m(x)$ is some multilinear polynomial in the $x$-variables (note that here we include all multilinear monomials and not only the set-multilinear ones). We show that for any $m$ that is set-multilinear on $w|_{N_w}$ the leading monomial of $g_m(x)$ is the set-multilinear monomial $m(\sigma(m)|_{A_w^{P_w}})$. This is where we focus on the submatrix of set-multilinear monomials within the bigger matrix of all multilinear monomials, as depicted in \autoref{fig:submatrix}. To prove this we need the following claim.

\begin{claim}\label{cla:main-rank-claim}
                        For any monomial $m$  set-multilinear on some $w|_{T}$, where $T\subseteq N_w$,
                        the leading monomial of $g_m(x)$   is less or equal to
                        \[m\left(\sigma(m)|_{A_w^S}\right),\] where $S$ is the maximal subset of $P_w$ such that $A_w^S \subseteq B_w^T$.

                        Moreover if $m$ is set-multilinear on $w|_{N_w}$, then the leading monomial of $g_m(x)$ \emph{ equals}
                        \[m\left(\sigma(m)|_{A_w^{P_w}}\right).\]
                \end{claim}
        \begin{proof}[Proof of Claim.]
                Proof by induction on the size of $T$. 

\Base If $T = \emptyset$, the only monomial set-multilinear on $w|_\emptyset$ is the empty monomial $1$. Now consider the partial assignment $\tau_1$ that maps all the $\vy$-variables to $0$. Now
                $\tau_1(f) = g_1(x)$, where $g_1(x) $   is the coefficient of the empty monomial $1$. On the other hand, since
                \[
                f = \frac{1}{\ks_w}\text{\quad over Boolean assignments,}
                \]
                we have that $\tau_1(f) = 1/-\beta$ over Boolean assignments. As $g_1(x)$ is multilinear, $g_1(x) = 1/-\beta$ as a polynomial identity and so the leading monomial of $g_1(x)$ is the empty monomial $1$.

\induction                 
                Suppose then that $T$ is non-empty, and let $m$ be a set-multilinear monomial on $w|_{T}$. Now consider the partial assignment $\tau_m$ that maps any variable in $m$   to $1$   and any other $\vy$-variable to $0$. By \autoref{eq:fgm}
\begin{equation}\label{eq:mfgprime}
                \tau_m(f) = \sum_{m'}g_{m'}(x),
\end{equation}               
                where $m'$ ranges over all submonomials of $m$. On the other hand, by the construction of $\ks_w$
                \[
                \tau_m(f) = \frac{1}{\sum x^{(i)}_\sigma -\beta}\text{\quad over Boolean assignments,}
                \]
                where the sum in the denominator ranges over those $i$ and $\sigma$ such that $A_w^{(i)}\subseteq B_w^T$ and $\sigma$  agrees with $\sigma(m)$  on the interval $A_w^{(i)}$. Note that for any $i\in P_w$ there is at most one $\sigma$ such that $x^{(i)}_\sigma$  appears in the sum (because $A_w^{(i)}\subseteq B_w^T$, and by construction of $\ks_w$, $x_\sigma^{(i)}$ for a fixed $\sigma$ is multiplied in $\ks_w$ by all products $\Pi_j y^{(j)}_{\rho_j}$ such that the concatenation of the $\rho_j$'s extends the string $\sigma$; but since we assigned 0-1 to all the  $\vy$-variables there is  a \emph{single }such concatenation, induced by our assignment of 1's; see \autoref{fig:matching}).

\begin{figure}
               \center
        \includegraphics[scale=0.8]{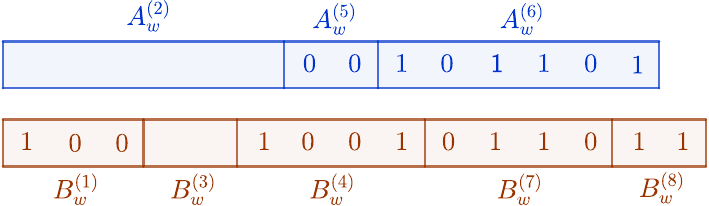}
         \caption{\small In this example, $T = \{1,4,7,8\} \subseteq N_w$ and  $m = y_{100}^{(1)} \cdot y_{1001}^{(4)} \cdot y_{0110}^{(7)} \cdot y_{11}^{(8)}$ is a set-multilinear monomial over $w \vert_T$. As $S = \{5,6\}$ is the maximal subset of $P_w$ such that $A_w^S \subseteq B_w^T$, we therefore have that the leading monomial of $g_m(x)$ (from \autoref{eq:fgm}) is less than or equal to $x_{00}^{(5)} \cdot x_{101101}^{(6)}$. Moreover, in the polynomial $\textbf{ks}_w$, the partial assignment setting the $y$-variables in $m$ to $1$ and the remaining $y$-variables to $0$ results in the polynomial $x_{00}^{(5)} + x_{101101}^{(6)} - \beta$.}
        \label{fig:leading-monomial}
\end{figure}

                It follows by \autoref{lem:FSTW21-deg-lb} that the leading monomial of $\tau_m(f) $  is the product of all the $x^{(i)}_\sigma$ appearing in the sum in the denominator above, and thus the leading monomial equals
\begin{equation}\label{eq:msigm}
                m\left(\sigma(m)|_{A_w^S}\right),
\end{equation}
                where $S$ is the maximal subset of $P_w$  such that $A_w^S\subseteq B_w^T$.
 By \autoref{eq:mfgprime}, this means that either the leading monomial of $g_m(x)$ is less than or equal to  \eqref{eq:msigm}, or otherwise is greater than  \eqref{eq:msigm} but is cancelled out in \eqref{eq:mfgprime} by  some monomial in $g_{m'}(x)$ for $m'$ a proper submonomial of $m$. 
But by induction assumption, for a \emph{proper} submonomial $m'$ of $m$ with $T'\subsetneq T$ and $m'$ set-multilinear on $w|_{T'}$, the  leading monomial of $g_{m'}(x)$ is less or equal to 
$m\left(\sigma(m')|_{A_w^{S'}}\right)$, where $S'$ is the maximal subset of $P_w$ such that $A_w^{S'} \subseteq B_w^{T'}$. Since, $m'$ is a submonomial of $m$ the monomial $m\left(\sigma(m')|_{A_w^{S'}}\right)$ is less than or equal to 
$m\left(\sigma(m)|_{A_w^{S}}\right)$, and so the above mentioned cancellation cannot occur, and we conclude that the leading monomial of $g_m(x)$ is less than or equal to \eqref{eq:msigm}.

%a possibly smaller (in the monomial ordering)  $m(\sigma(m)|_{A_w^S}).$ \iddo{Perhaps inaccurate here: because the leading monomial of $g_{m'}(x)$ is less or equal to a *possibly smaller monomial than $m(\sigma(m)|_{A_w^S})$, and not $m(\sigma(m)|_{A_w^S})$ itself, and this should be stated, or explained while changing the induction statement (i.e., the statement of the claim).}

%\iddo{original text:                As this is the leading monomial of $\tau_m(f)$, it follows that any monomial on $g_m(x)$ is also less or equal to $m(\sigma(m)|_{A_w^S})$.} 
\bigskip 

It remains to show that the leading monomial of $g_m(x)$ is \emph{precisely} \eqref{eq:msigm}, when $m$ is a set-multilinear monomial on $w|_{N_w}$. Let $m'$ be a proper submonomial of $m$  that is set-multilinear over $w|_T$ for some $T\subsetneq N_w$. By the assumption that $w$  is balanced there is some $i\in P_w$  such that $A_w^{(i)} \nsubseteq B_w^T$, and thus the leading monomial of $g_{m'}(x)$  is properly smaller than $
                m(\sigma(m)|_{A_w^{P_w}})
                $. Hence, by \autoref{eq:msigm} (and the sentence preceding it) the leading monomial of $g_m(x)$   must equal $
                m(\sigma(m)|_{A_w^{P_w}})
                $.
        \end{proof}

        By the claim for any $m$ that is set-multilinear on $w|_{N_w}$ the leading monomial of $g_m(x)$ is the monomial $m(\sigma(m)|_{A_w^{P_w}})$, and these include all the set-multilinear monomials on $w|_{P_w}$. Thus the column space of $M_w(f)$ spans the space of all set-multilinear polynomials on $w|_{P_w}$, and $M_w(f)$ is full-rank (where a column in $M_w(f)$ determines a polynomial that is a linear combination of the positive monomials in its rows).

        \end{proof}

    \begin{corollary}\label{cor:relrk-lower-bound}
        Let $w\in\ZZ^d$  be a balanced word with $|w_i|\leq b$  for all $i\in [d]$, and let $f$   be the multi-linear polynomial  so that
        \[
        f = \frac{1}{\ks_w} \text{  over Boolean assignments}.
        \]
        Then $\relrk_w(f)\geq 2^{-b/2}$.
    \end{corollary}

        \begin{proof}
                Assume again without a loss of generality that $|w_{N_w}| \geq |w_{P_w}|$. By the ``balanced-ness'' and the assumption that $|w_i|\leq b$  for all $i\in [d]$, we know that $|w_{P_w}| - |w_{N_w}| \geq -b$. By \autoref{lemma:rank-lower-bound}, $M_w(f)$  is of rank $|M^P_w|$, and so
                \[
                \relrk_w(f) = \sqrt{\frac{|M^P_w|}{|M^N_w|}} = \sqrt{2^{|w_{P_w}| - |w_{N_w}|}} \geq 2^{-b/2}.
                \]
        \end{proof}

\subsection{Lower Bound for Constant-Depth Set-Multilinear Circuits}

In this section we prove the following lower bound on bounded-depth set-multilinear circuits.

\begin{lemma}\label{lemma:sml-lower-bound}
        Let $d,k,\Delta\in \NN_+$ with $k \geq 10d$. Let $w\in\ZZ^d$  be a balanced word in the vocabulary $\{\ak,-k\}$, where $\alpha = \nicefrac{1}{\sqrt{2}}$, and let $f$  be the multilinear polynomial which equals $\nicefrac{1}{\ks_w}$ over Boolean assignments. Then any set-multilinear circuit of product-depth $\Delta$ computing the set-multilinear projection $\Pi_w(f)$ has size at least
        \[2^{k\left(\frac{d^{\nicefrac{1}{(2^\Delta - 1)}}-20}{40\Delta}\right)} \]
\end{lemma}
\begin{proof}
We prove the lemma by using the following claim from \cite{LST21}.

\begin{claim}[\cite{LST21} Claim 16]\label{claim:relrk-upper-bound}
        Let $k \geq 10d$. Let $w$ be any word of length $d$ with entries in $\{\ak,-k\}$, where  $\alpha = \nicefrac{1}{\sqrt{2}}$. Then for any $\Delta\geq 1$, any set-multilinear formula $C$ of product-depth $\Delta$ of size at most $s$ satisfies
        \[
        \relrk_w(C) \leq s\,\cdot\, 2^{\frac{-kd^{\nicefrac{1}{(2^\Delta - 1)}}}{20}}
        \]
\end{claim}

Let $C$ be a set-multilinear circuit of size $s$  and product-depth $\Delta$  computing $\Pi_w(f)$. We can transform $C$ into a set-multilinear formula $F$ of size $s^{2\Delta}$  and product-depth $\Delta$  computing $\Pi_w(f)$.

Now by \autoref{lemma:rank-lower-bound} and \autoref{claim:relrk-upper-bound}, we have that
\[
2^{-k}\leq \relrk_w(\Pi_w(f))\leq s^{2\Delta} \,\cdot\, 2^{\frac{-kd^{\nicefrac{1}{(2^\Delta - 1)}}}{20}},
\]
and thus
\[
s^{2\Delta}\geq 2^{k\left(\frac{d^{\nicefrac{1}{(2^\Delta - 1)}}-20}{20}\right)},
\]
from which the claim follows.
\end{proof}

\subsection{Lower Bound for Constant-Depth Circuits}

Finally in this section we prove  \autoref{theorem:hardz-lower-bound}.
To prove this theorem we reduce the task of computing $f$ in \autoref{theorem:hardz-lower-bound} to computing the set-multilinear projection $\Pi_w(f')$ of the multilinear $f'$  that equals $\nicefrac{1}{\ks_w}$ over the Boolean values for a suitable word $w$. For this we require the following lemma which can be proved in a manner similar to Proposition 9 in \cite{LST21}.

\begin{lemma}\label{lemma:set-multilinearization}
        Let $s,N,d$ and $\Delta$  be growing parameters with $s\geq Nd$. Assume that $\chara(\FF) = 0$ or $\chara(\F) > d$. Let $C$  be a circuit of size at most $s$  and product-depth at most $\Delta$ computing a polynomial $f$. Let $\overline{X} = \langle X_1,\ldots,X_d\rangle$  be a sequence of pairwise disjoint sets of variables, each of size at most $N$. Then there is a set-multilinear circuit $\tilde{C}$ of size at most $d^{O(d)}\poly(s)$ and product-depth at most $2\Delta$ computing the set-multilinear projection $\Pi_{\overline{X}}(f)$ of $f$.
\end{lemma}

With this lemma at hand, we are ready to prove Theorem \ref{theorem:hardz-lower-bound}.

\begin{proof}[Proof of \autoref{theorem:hardz-lower-bound}]
        Let $C$  be a circuit of size at most $s\geq n$  and product-depth $\Delta$  computing $f$. Now let $k = \lfloor \nicefrac{\log n}{2}\rfloor$ and $d = \lfloor \nicefrac{\log n}{25}\rfloor$, and note that $d2^{k} < n\leq s$ for large enough $n$. Note also that $k\geq 10d$ for large enough $n$. Let $w$ be a balanced word of length $d$ on the alphabet $\{\ak,-k\}$. One can easily construct such a word by induction on $d$.

%       Let $w'$  be a word on the alphabet $\{\ak,-k\}$  of length $d_0$  that satisfies $|\sum_i w'_i| < k$. We can construct such word easily by induction on $d_0$. If $\sum_{i\in [d_0]} w'_i \leq 0$, let $w = w'$, otherwise let $w = \langle-k\rangle^\frown w'$, where ``${}^\frown$'' denotes concatenation. Let $d$ be the length of $w$, i.e., $d = d_0$  or $d = d_0 + 1$. Note that we still have $|\sum_{i\in [d]} w_i| < k$. Note also that $0 \leq \sum_{i\in [d]}w_i < k$, and thus $w$  is balanced. Moreover for any $A_i$  there are at most two $B_j$'s such that $A_i\cap B_j\neq \emptyset$. Thus the products in the expression \ref{lift} involve at most two sums.

        Note that the polynomial $\ks_w$ is now of degree at most $4$, as any $B_w^{(j)}$ overlaps at most $3$ different $A_w^{(i)}$'s and any $A_w^{(i)}$ overlaps at most $2$ different $B_w^{(j)}$'s. Also, by choice of parameters, $\ks_w$ involves less than $n$  many variables. Hence there is some partial assignment $\tau_w$ to the variables $\{z_{ijk\ell},x_i : i,j,k,\ell\in [n]\}$ that maps $\hardzp$ to the polynomial $\ks_w$ (up to renaming of variables). By applying this partial mapping to $C$, we obtain a circuit $C'$  of size at most $s$  and product-depth $\Delta$  that computes the multilinear $f'$  that equals $\nicefrac{1}{\ks_w}$ over Boolean assignments. Now, by Lemma \ref{lemma:set-multilinearization}, there is a set-multilinear circuit $C'$ of size $d^{O(d)}\poly(s)$ and product-depth $2\Delta$ computing the set-multilinear projection $\Pi_w(f')$ of $f'$.

        By \autoref{lemma:sml-lower-bound} any set-multilinear circuit of product-depth $2\Delta$  computing $\Pi_w(f')$ has size at least
        \[2^{k\left(\frac{d^{\nicefrac{1}{(2^{2\Delta} - 1)}}-20}{80\Delta}\right)}\geq n^\frac{d^{\nicefrac{1}{(2^{2\Delta} - 1)}}-20}{200\Delta}.\]
        Putting everything together we have that
        \[d^{O(d)}\poly(s)\geq n^\frac{d^{\nicefrac{1}{(2^{2\Delta} - 1)}}-20}{200\Delta}.\]

Now, given that $\Delta \leq \nicefrac{1}{4}\log\log\log n$, by the choice of $d$, we have that $2\Delta\leq \nicefrac{1}{2}\log\log 30d$. 
%Now if $2\Delta\geq \frac{1}{2}\log\log d$, the stated lower bound becomes trivial. On the other hand if $2\Delta < \frac{1}{2}\log\log d$, then 
Then for large enough $n$ we have that
\[n^{\frac{d^{\nicefrac{1}{(2^{2\Delta} -1)}}-20}{200\Delta}} \geq n^{\frac{2^{\sqrt{\log (d/30)}/2 }}{50\log\log 30d}},\]
and as
\[n^{\frac{2^{\sqrt{\log (d/30)}/2}}{50\log\log 30d}} \geq d^{\omega(d)},\]
the lower bound follows.

\end{proof}

\noindent\textbf{Comment:} We remark that if we make sure that the word $w$  above leans towards the negative monomials, meaning that $|w_{N_w}|\geq |w_{P_w}|$, we can actually prove the lower bound for the following degree-$4$ variant of subset-sum
\[
\sum_{i,j,k\in [n]} z_{ijk}x_ix_jx_k - \beta.
\]
We have however opted for the proof above for simplicity, as it works for any balanced word over the vocabulary $\{\ak,-k\}$ and does not involve any set-up of a suitable word that could distract from the main idea.
% 

% 
% %%%%%%%%%%%%%%%%%%%%%%%%%%%%%%%%%%%%%%%%%%%%%
% %%%%%%%%%%%%%%%%%%%%%%%%%%%%%%%%%%%%%%%%%%%%%
 \section{Conclusions and Open Problems}
% %%%%%%%%%%%%%%%%%%%%%%%%%%%%%%%%%%%%%%%%%%%%%
% %%%%%%%%%%%%%%%%%%%%%%%%%%%%%%%%%%%%%%%%%%%%%
% 
% 
% 

The main goal of this work is to advance on the frontiers of strong propositional proofs lower bounds. We provide the first lower bounds against  algebraic proof systems operating with constant-depth circuits, where the hard instance is computable itself with  small constant-depth circuits. 
Our hard instances are combinatorial, simple, and have 0-1 coefficients, and the lower bounds work in the unit-cost model of algebraic circuits, namely, where the size does not depend on the magnitude of coefficients used in the polynomials appearing in the proof 

%(unlike for example \cite{Ale21} that used bit-complexity  measure for proofs as well as hard instances with large coefficients; cf.~\cite{AGHT20}).

Thus, our result brings us to the natural and standard setting of proof complexity lower bounds, while coming closer to  CNF hard instances (since the magnitude of coefficients in the hard instances do not play a role in our lower bound proofs). 

On the other hand, establishing lower bounds against CNF formulas in strong algebraic proof systems stays a remarkable open problem, since it necessitates a lower bound technique that is different from  the functional lower bound approach we used or the bit-complexity/large-coefficients approach of \cite{Ale21,AGHT20} (or at least a substantially modified technique than those two techniques). Note that such lower bounds would also imply constant-depth Frege with counting gates ($\ACZ[p]$-Frege) lower bounds, which is an important  long-standing open problem in proof complexity. This leads us to the following set of \textit{open problems}:

% For the for the CR programme it is not clear how %AF22 is part of it because of their weakness.
% - So we rectify this.
% - For CR Alekseev is okay, but weaker models.
% - Major goal is to say something about CNF/Frege proofs:
% -- for this AF22 and Alekseev doesn't work.
% --- AF22 because immediately it fails; and Alekseev
%  because it's impossible to get a CNF as shown in
%   AGHT20 explanation. Because of the large coefficients.Comment: Alekseev's
% proof system is weaker in the sense that coefficients
% ould have been written more compactly as straight
% line programs, while he writes them as binary strings.
% 
% Hence, over all we rectify all the problems above, while paying with multilinearity
% of the refutation:
%  close to the CR program and closer to the CNF Frege
%   lower bound goal ; except for the problem of using CNF.
%  This remains  a great obstacle and would necessitate a new idea
%   in proof complexity (would also imply the technical
%    breakthrough of AC0-p-Frege lower bounds which
%     is still open), so this is bound to be a hard problem.
%\end{verbatim}

\begin{enumerate}
\item \textbf{CNF hard instances}: 
Can we establish lower bounds against  strong algebraic proof systems, and specifically constant-depth IPS proofs for a family of CNF formulas? As mentioned above, this is a very challenging problem, or at least one with important consequences in proof complexity. 

% remember the last enum item number
\setcounter{enumresume}{\value{enumi}}
\end{enumerate} 

Before tackling this difficult open problem, we identify several possibly less challenging ones that seem to be prerequisites for solving open problem 1 (at least as far as taking the proof complexity lower bound approach in the current paper):

\begin{enumerate}
% resume the last enum item number
\setcounter{enumi}{\value{enumresume}}

\item \textbf{Finite fields}: Can we establish lower bounds against strong algebraic proof systems, and specifically constant-depth IPS proofs over finite fields? Both the functional lower bound argument from Forbes \textit{et al.}~\cite{FSTW21} and the Limaye \textit{et al.}~\cite{LST21} technique, use the fact that the fields are sufficiently big, or have characteristic 0. For the former approach, characteristic 0 fields are  essential: first, the subset sum instance is not necessarily unsatisfiable over finite fields. But more crucially, the whole functional lower bound approach for IPS hinges on lower bounding the size of a circuit computing the function $\nicefrac{1}{f}$ over the Boolean cube, for some efficiently computable polynomial $f$. However, over finite fields $\nicefrac{1}{f}$ is efficiently computed, when $f$ is, over the Boolean cube. For the latter \cite{LST21} approach, large fields do not seem to be as crucial to the argument (it is used only in  the homogenization procedure to yield low-depth circuits, using polynomial interpolation, the latter requires a sufficiently large field; this homogenization is based on a generalization of Shpilka and Wigderson \cite{SW01}). However, there may be a different way to homogenize constant-depth circuits without increasing too much the depth even over finite fields. 

\item \textbf{No multilinear requirement}: Can we get rid of the requirement for multilinearity of the IPS refutations in our lower bounds? Namely, can we use a stronger proof system than multilinear \lbIPS? We discussed this requirement in the introduction. It comes from the requirement in the Limaye \textit{et al}.~\cite{LST21} technique to consider set-multilinear polynomials, as well as the use of the functional lower bound approach from \cite{FSTW21} which focuses on functions computed on the  Boolean cube alone. A hard set-multilinear polynomial can compute the same function over the Boolean cube as a polynomial whose set-multilinear projection is in fact zero (and hence easy to compute), which breaks our argument. It is unclear at the moment how to overcome this obstacle. This leads us to the following interesting  question in algebraic circuit complexity proper:      

\item \textbf{Functional lower bounds for constant-depth circuits}: Can we prove a lower bound on the size of every constant-depth circuit computing a certain \emph{function} (in contrast to a certain specific polynomial; this is the difference between ``semantic'' lower bounds, and the weaker notion of ``syntactic'' lower bounds in algebraic circuit complexity)?   
\end{enumerate}

%\paragraph*{Acknowledgments}
%sss

%%%%%%%%%%%%%%%%%%%%%%%%%%%%%%%%%%%%%%%%%%%%
%%%%%%%%%%%%%%%%%%%%%%%%%%%%%%%%%%%%%%%%%%%%
\bibliographystyle{plain}
\bibliography{PrfCmplx-Bakoma}

\begin{thebibliography}{10}

\bibitem{Ale21}
Yaroslav Alekseev.
\newblock A lower bound for polynomial calculus with extension rule.
\newblock In Valentine Kabanets, editor, {\em 36th Computational Complexity
  Conference, {CCC} 2021, July 20-23, 2021, Toronto, Ontario, Canada (Virtual
  Conference)}, volume 200 of {\em LIPIcs}, pages 21:1--21:18. Schloss Dagstuhl
  - Leibniz-Zentrum f{\"{u}}r Informatik, 2021.

\bibitem{AGHT20}
Yaroslav Alekseev, Dima Grigoriev, Edward~A. Hirsch, and Iddo Tzameret.
\newblock Semi-algebraic proofs, {IPS} lower bounds, and the
  {\(\tau\)}-conjecture: can a natural number be negative?
\newblock In {\em Proceedings of the 52nd Annual {ACM} {SIGACT} Symposium on
  Theory of Computing, {STOC} 2020}, pages 54--67. {ACM}, 2020.

\bibitem{AF22}
Robert Andrews and Michael~A. Forbes.
\newblock Ideals, determinants, and straightening: Proving and using lower
  bounds for polynomial ideals.
\newblock In {\em 54th Annual {ACM} {SIGACT} Symposium on Theory of Computing,
  {STOC} 2022}, 2022.

\bibitem{BeameIKPP96}
Paul Beame, Russell Impagliazzo, Jan Kraj{\'{\i}}{\v{c}}ek, Toniann Pitassi,
  and Pavel Pudl{\'a}k.
\newblock Lower bounds on {H}ilbert's {N}ullstellensatz and propositional
  proofs.
\newblock {\em Proc. London Math. Soc. (3)}, 73(1):1--26, 1996.

\bibitem{BussIKPRS96}
Samuel~R. Buss, Russell Impagliazzo, Jan Kraj{\'{\i}}{\v{c}}ek, Pavel
  Pudl{\'{a}}k, Alexander~A. Razborov, and Ji{\v{r}}{\'{\i}} Sgall.
\newblock Proof complexity in algebraic systems and bounded depth {F}rege
  systems with modular counting.
\newblock {\em \ComputationalComplexity}, 6(3):256--298, 1996.

\bibitem{CEI96}
Matthew Clegg, Jeffery Edmonds, and Russell Impagliazzo.
\newblock Using the {G}roebner basis algorithm to find proofs of
  unsatisfiability.
\newblock In {\em Proceedings of the 28th Annual ACM Symposium on the Theory of
  Computing (Philadelphia, PA, 1996)}, pages 174--183, New York, 1996. ACM.

\bibitem{CoxLittleOShea07}
David Cox, John Little, and Donal O'Shea.
\newblock {\em Ideals, varieties, and algorithms}.
\newblock Undergraduate Texts in Mathematics. Springer, New York, third
  edition, 2007.
\newblock An introduction to computational algebraic geometry and commutative
  algebra.

\bibitem{FSTW21}
Michael~A. Forbes, Amir Shpilka, Iddo Tzameret, and Avi Wigderson.
\newblock Proof complexity lower bounds from algebraic circuit complexity.
\newblock {\em Theory Comput.}, 17:1--88, 2021.

\bibitem{GP18}
Mika G{\"{o}}{\"{o}}s and Toniann Pitassi.
\newblock Communication lower bounds via critical block sensitivity.
\newblock {\em {SIAM} J. Comput.}, 47(5):1778--1806, 2018.

\bibitem{GH03}
Dima Grigoriev and Edward~A. Hirsch.
\newblock Algebraic proof systems over formulas.
\newblock {\em Theoret. Comput. Sci.}, 303(1):83--102, 2003.
\newblock Logic and complexity in computer science (Cr\'eteil, 2001).

\bibitem{GP14}
Joshua~A. Grochow and Toniann Pitassi.
\newblock Circuit complexity, proof complexity, and polynomial identity
  testing: The ideal proof system.
\newblock {\em J. {ACM}}, 65(6):37:1--37:59, 2018.

\bibitem{IMP20}
Russell Impagliazzo, Sasank Mouli, and Toniann Pitassi.
\newblock The surprising power of constant depth algebraic proofs.
\newblock In Holger Hermanns, Lijun Zhang, Naoki Kobayashi, and Dale Miller,
  editors, {\em {LICS} '20: 35th Annual {ACM/IEEE} Symposium on Logic in
  Computer Science, Saarbr{\"{u}}cken, Germany, July 8-11, 2020}, pages
  591--603. {ACM}, 2020.

\bibitem{IMP22}
Russell Impagliazzo, Sasank Mouli, and Toniann Pitassi.
\newblock Lower bounds for polynomial calculus with extension variables over
  finite fields.
\newblock {\em Electron. Colloquium Comput. Complex.}, 2022.

\bibitem{IPS99}
Russell Impagliazzo, Pavel Pudl{\'{a}}k, and Ji{\v{r}}{\'{\i}} Sgall.
\newblock Lower bounds for the polynomial calculus and the {G}r{\"{o}}bner
  basis algorithm.
\newblock {\em Computational Complexity}, 8(2):127--144, 1999.

\bibitem{Kno17}
Alexander Knop.
\newblock Ips-like proof systems based on binary decision diagrams.
\newblock {\em Electron. Colloquium Comput. Complex.}, page 179, 2017.

\bibitem{LTW18}
Fu~Li, Iddo Tzameret, and Zhengyu Wang.
\newblock Characterizing propositional proofs as noncommutative formulas.
\newblock In {\em SIAM Journal on Computing}, volume~47, pages 1424--1462,
  2018.
\newblock Full Version: \url{http://arxiv.org/abs/1412.8746}.

\bibitem{LST21}
Nutan Limaye, Srikanth Srinivasan, and S{\'{e}}bastien Tavenas.
\newblock Superpolynomial lower bounds against low-depth algebraic circuits.
\newblock In {\em 62nd {IEEE} Annual Symposium on Foundations of Computer
  Science, {FOCS} 2021, Denver, CO, USA, February 7-10, 2022}, pages 804--814.
  {IEEE}, 2021.

\bibitem{Pit97}
Toniann Pitassi.
\newblock Algebraic propositional proof systems.
\newblock In {\em Descriptive complexity and finite models (Princeton, NJ,
  1996)}, volume~31 of {\em DIMACS Ser. Discrete Math. Theoret. Comput. Sci.},
  pages 215--244. Amer. Math. Soc., Providence, RI, 1997.

\bibitem{Pit98}
Toniann Pitassi.
\newblock Unsolvable systems of equations and proof complexity.
\newblock In {\em Proceedings of the International Congress of Mathematicians,
  Vol. III (Berlin, 1998)}, number Vol. III, pages 451--460, 1998.

\bibitem{PT16}
Tonnian Pitassi and Iddo Tzameret.
\newblock Algebraic proof complexity: Progress, frontiers and challenges.
\newblock {\em ACM SIGLOG News}, 3(3), 2016.

\bibitem{RT07}
Ran Raz and Iddo Tzameret.
\newblock Resolution over linear equations and multilinear proofs.
\newblock {\em Ann. Pure Appl. Logic}, 155(3):194--224, 2008.

\bibitem{RT06}
Ran Raz and Iddo Tzameret.
\newblock The strength of multilinear proofs.
\newblock {\em Computational Complexity}, 17(3):407--457, 2008.

\bibitem{Razb98}
Alexander~A. Razborov.
\newblock Lower bounds for the polynomial calculus.
\newblock {\em Comput. Complexity}, 7(4):291--324, 1998.

\bibitem{Sap17-survey}
Ramprasad Saptharishi.
\newblock A survey of lower bounds in arithmetic circuit complexity, 2016-2022.
\newblock \url{https://github.com/dasarpmar/lowerbounds-survey/releases}.

\bibitem{SW01}
Amir Shpilka and Avi Wigderson.
\newblock Depth-3 arithmetic circuits over fields of characteristic zero.
\newblock {\em Comput. Complexity}, 10:1--27, 2001.

\bibitem{SY10}
Amir Shpilka and Amir Yehudayoff.
\newblock Arithmetic circuits: A survey of recent results and open questions.
\newblock {\em Foundations and Trends in Theoretical Computer Science},
  5(3-4):207--388, 2010.

\bibitem{Tza11-I&C}
Iddo Tzameret.
\newblock Algebraic proofs over noncommutative formulas.
\newblock {\em Inf. Comput.}, 209(10):1269--1292, 2011.

\end{thebibliography}

%\newpage
%\begin{center}\small{--- Page left blank for ECCC stamp ---}\end{center}

% that's all folks
\end{document}